\documentclass[11pt,letterpaper]{article}
\usepackage[letterpaper,margin=1in]{geometry}
\usepackage[T1]{fontenc}
\usepackage{lmodern}
\usepackage{authblk}
\usepackage{amsmath,amssymb,amsthm}
\usepackage{mathtools,bm}
\usepackage{graphicx,booktabs,array,tabularx,multirow}
\usepackage{xcolor}
\usepackage[numbers,sort&compress]{natbib}
\usepackage{tikz}
\usetikzlibrary{arrows.meta,positioning,calc,fit}
\usepackage[colorlinks=true,linkcolor=blue,citecolor=blue,urlcolor=blue]{hyperref}
\hypersetup{pdfcreator={pdfLaTeX}}
\allowdisplaybreaks[1]

\def\TheoremName{Theorem}\def\LemmaName{Lemma}
\def\PropositionName{Proposition}\def\CorollaryName{Corollary}
\def\DefinitionName{Definition}\def\RemarkName{Remark}
\def\ClaimName{Claim}\def\AlgorithmName{Algorithm}
\def\InputName{Input}\def\OutputName{Output}
\theoremstyle{plain}
\newtheorem{theorem}{\TheoremName}[section]
\newtheorem{lemma}[theorem]{\LemmaName}
\newtheorem{proposition}[theorem]{\PropositionName}

\theoremstyle{definition}

\theoremstyle{remark}

\newcounter{algorithm}
\newcounter{algline}
\newenvironment{algorithm}[1][]{%
	\begin{figure}[!htbp]\small
		\refstepcounter{algorithm}
		\hrule height .5pt\vspace{2pt}
		\noindent\textbf{\AlgorithmName~\thealgorithm}\enspace #1\par
		\vspace{2pt}\hrule height .3pt\vspace{3pt}
		\begin{list}{\arabic{algline}:}{%
				\usecounter{algline}\setlength{\leftmargin}{1.8em}%
				\setlength{\labelwidth}{1.3em}\setlength{\labelsep}{.5em}%
				\setlength{\rightmargin}{0pt}\setlength{\itemsep}{0pt}%
				\setlength{\parsep}{0pt}\setlength{\topsep}{0pt}%
				\setlength{\partopsep}{0pt}}%
		}{\end{list}\vspace{3pt}\hrule height .3pt\end{figure}}
\newcommand{\State}[1]{\item #1}
\newcommand{\StateIndent}[2]{\item\hspace*{#1\dimexpr1em\relax}#2}
\newcommand{\Input}[1]{\State{\textbf{\InputName:} #1}}
\newcommand{\Output}[1]{\State{\textbf{\OutputName:} #1}}

\newcommand{\supp}{\operatorname{supp}}
\newcommand{\OPT}{\operatorname{OPT}}
\newcommand{\E}{\mathbb{E}}
\newcommand{\R}{\mathbb{R}}
\newcommand{\one}{\mathbf{1}}
\newcommand{\zero}{\mathbf{0}}
\newcommand{\pos}[1]{[#1]_{+}}
\newcommand{\Top}{\operatorname{top2}}
\newcommand{\ceil}[1]{\left\lceil #1\right\rceil}
\newcommand{\floor}[1]{\left\lfloor #1\right\rfloor}
\newcommand{\comb}{\textsuperscript{\normalfont **}}
\newcommand{\cost}{\mathbf{c}}
\newcommand{\prices}{\mathbf{p}}
\newcommand{\hprices}{\widehat{\mathbf{p}}}
\newcommand{\res}{\mathbf{r}}
\newcommand{\ghost}{\mathbf{m}}
\newcommand{\lvec}{\boldsymbol{\lambda}}

\title{A Combinatorial $(1/2-\varepsilon)$-Approximation for $k$-Submodular Maximization under Multiple Knapsacks}
\author[1]{Siyuan Chen}
\author[1,2]{Suixiang Gao}
\author[1]{Zheyu Jiang}
\author[1]{Wenguo Yang}
\affil[1]{School of Mathematical Sciences,
	University of Chinese Academy of Sciences,
	Beijing 100049, China}
\affil[2]{Zhongguancun Laboratory,
	Beijing 100000, China}
\date{\vspace{-1.5em}
	\small\texttt{%
		chensiyuan221@mails.ucas.ac.cn, sxgao@ucas.ac.cn,\\
		jiangzheyu24@mails.ucas.ac.cn, yangwg@ucas.ac.cn}}
\hypersetup{pdftitle={A Combinatorial (1/2-epsilon)-Approximation for Non-Monotone k-Submodular Maximization under Multiple Knapsacks}}
\hypersetup{pdfauthor={Siyuan Chen; Suixiang Gao; Zheyu Jiang; Wenguo Yang}}
\begin{document}
	\hypersetup{pageanchor=false}
	\maketitle
	\begingroup
	\renewcommand{\thefootnote}{}
	\makeatletter
	\long\def\@makefntext#1{\noindent#1}
	\makeatother
	\begin{NoHyper}
		\footnotetext{The authors are listed in alphabetical order.}
	\end{NoHyper}
	\endgroup
	\thispagestyle{empty}
	\vspace{1.5em}
	\begin{abstract}
		We give a randomized combinatorial \((1/2-\varepsilon)\)-approximation for nonnegative non-monotone \(k\)-submodular maximization under a fixed number \(d\) of knapsack constraints, for \(k\ge2\). The algorithm selects whole elements and follows the proportional top-2 rule of Chen et al. \cite{chen2026}: it samples from the two labels with the largest marginal gains, with probabilities proportional to the positive parts of these gains. Extending their single-knapsack Bellman analysis, we construct a common resource potential and associated resource weights, with potential losses charged to greedy gains. Since these weights depend on the unknown optimum, we discretize their defining coefficients and enumerate monotone guesses over phases determined by resource consumption. This yields an algorithm using at most
		$\mathcal{A}(d,\varepsilon)(nk+1)^{H+2}$ value queries,
		where $n$ is the number of elements,
		$\mathcal{A}(d,\varepsilon)$ depends only on $d$ and $\varepsilon$,
		and $H=O(d\log(2d)/\varepsilon^2)$. The approximation ratio can be made arbitrarily close to the optimal oracle threshold for \(k=2\) and is asymptotically tight as \(k\to\infty\).
	\end{abstract}
	\clearpage
	\hypersetup{pageanchor=true}
	\pagenumbering{arabic}
	\section{Introduction}\label{sec:intro}
	A budgeted labeling problem asks which elements to select and which of $k$ labels to assign to them.
	A $k$-submodular objective combines diminishing returns within a consistent labeling with a compatibility condition between different labels.
	This framework extends submodularity and bisubmodularity and includes applications such as multi-type sensor placement and influence allocation~\cite{huberkolmogorov2012,ohsakayoshida2015}.
	When elements consume several kinds of resources, their selected support must satisfy several knapsack constraints at once.

	For $k\ge2$, we study nonnegative, possibly non-monotone $k$-submodular maximization under a fixed number of knapsack constraints.
	Chen et al.~\cite{chen2026} obtained a combinatorial $1/2$-approximation for a single knapsack using the proportional top-2 rule and a Bellman analysis of the remaining budget.
	We extend their approach to multiple knapsacks.
	Our algorithm uses resource weights to combine an element's resource consumptions into a single weighted cost for greedy selection, and our Bellman analysis accounts for all budgets simultaneously.
	
	\begin{theorem}[Main result, informal]\label{thm:informal}
		For every fixed number $d\ge1$ of knapsacks and every $0<\varepsilon<1/2$, there is a randomized combinatorial $(1/2-\varepsilon)$-approximation for nonnegative non-monotone $k$-submodular maximization, $k\ge2$.
		For fixed $d,\varepsilon$, the number of value queries is polynomial in the number of elements and labels.
		More precisely, it is at most
		\[
		\mathcal A(d,\varepsilon)(nk+1)^{\,O(d\log(2d)/\varepsilon^2)+2},
		\]
		where $n$ is the number of elements and $\mathcal A(d,\varepsilon)$ is independent of $n,k$.
	\end{theorem}
	
	The algorithm selects whole elements and follows Chen et al.'s proportional top-2 rule: among the two labels with largest marginal gains, it samples with probabilities proportional to their positive parts.
	It requires neither multilinear-extension evaluation nor fractional rounding.
	Its finite implementation enumerates small labeled initial solutions and guesses for the coefficients defining the resource weights.
	For each initial solution, the algorithm uses at most $\mathcal B(d,\varepsilon)(n^2k+1)$ value queries to explore all coefficient guesses and construct the resulting candidate solutions.
	Theorem~\ref{thm:main} and Proposition~\ref{prop:queries} give the precise guarantee and query count.
	The approximation approaches the optimal oracle threshold for $k=2$ and is asymptotically tight as $k\to\infty$, up to the arbitrarily small loss $\varepsilon$.
	
	\begin{table}[tbp]
		\centering\small
		\caption{Selected approximation results for $k$-submodular maximization under knapsack constraints, $k\ge2$.}
		\label{tab:knapsacks}
		\begin{tabular}{@{}p{.16\textwidth}p{.39\textwidth}p{.39\textwidth}@{}}
			\toprule
			Constraint & Monotone & Non-monotone\\
			\midrule
			Single knapsack &
			\begin{minipage}[t]{\linewidth}\vspace{0pt}
				$\tfrac12(1-e^{-1})$\comb\ \cite{tang2022}\\[2pt]
				$\tfrac12(1-e^{-2})$\comb\ \cite{xiao2026}\\[2pt]
				$1/2-\varepsilon$ \cite{zhou2025}\\[2pt]
				$1/2$\comb\ \cite{wang2025}
			\end{minipage}
			&\begin{minipage}[t]{\linewidth}\vspace{0pt}
				$\tfrac13(1-e^{-3})$\comb; $1/3-\varepsilon$ \cite{xiao2026}\\[2pt]
				$1/3-\varepsilon$ \cite{zhou2025}; $1/3$\comb\ \cite{wang2025}\\[2pt]
				$a_k-\varepsilon^\bigstar$ \cite{aggarwal2026}\\[2pt]
				$1/2$\comb\ \cite{chen2026}
			\end{minipage}\\[8pt]
			$d$ knapsacks &
			\begin{minipage}[t]{\linewidth}\vspace{0pt}
				$\tfrac1{2+2d}-\varepsilon$\comb\ \cite{gong2026}\\[2pt]
				$1/2-\varepsilon$ \cite{zhou2025}\\[2pt]
				\textbf{$1/2-\varepsilon$}\comb\ (\textbf{this work})
			\end{minipage}
			&\begin{minipage}[t]{\linewidth}\vspace{0pt}
				$\tfrac1{3+2d}-\varepsilon$\comb\ \cite{gong2026}\\[2pt]
				$1/3-\varepsilon$ \cite{zhou2025}\\[2pt]
				$a_k-\varepsilon^\bigstar$ \cite{aggarwal2026}\\[2pt]
				\textbf{$1/2-\varepsilon$}\comb\ (\textbf{this work})
			\end{minipage}\\
			\bottomrule
		\end{tabular}
		\par\vspace{4pt}
		\begin{minipage}{.97\textwidth}\footnotesize
			\comb\ denotes a combinatorial algorithm.\\
			$\bigstar$ $a_2=0.4456016834\ldots$ and $a_k=0.4541264676\ldots$ for $k\ge3$.
		\end{minipage}
	\end{table}
	
	\subsection{Technical overview}
	Our starting point is the single-budget Bellman analysis of Chen et al.~\cite{chen2026}.
	For the analysis, fix an optimal labeling.
	At each step, the comparison labeling agrees with the algorithm on the selected elements and retains the original labels of the remaining available optimum elements.
	Its value minus the current value is the \emph{comparison gap}.
	The proportional top-2 rule controls the expected change in this gap.
	The Bellman potential combines value already obtained with the positive part of the comparison gap, scaled by a resource-dependent coefficient.
	
	With several budgets, the coefficient and the selection rule must be designed together.
	We combine the single-budget coefficients through a smooth lower approximation to their minimum, obtaining one coefficient that respects every budget.
	The construction supplies resource weights, which turn an element's cost vector into one \emph{weighted cost}.
	Greedy selection maximizes expected gain divided by this cost.
	The weights satisfy two linked bounds: one controls the remaining comparison gap by gain density, and the other controls the coefficient's decrease by the selected element's weighted cost.
	That cost cancels when the bounds are combined, allowing the expected gain to cover the potential loss of a whole-element selection.
	
	The remaining obstacle is that an optimum element can become too large relative to an available budget.
	Removing it from both the candidates and the comparison labeling preserves the density argument but may lower the potential.
	We bound the total resulting potential loss on every realized run in terms of the largest remaining optimum marginal at the starting solution, independently of how many elements are deleted.
	A sufficiently long prefix, chosen greedily within the optimum, makes this marginal small relative to its own value.
	Retaining that value creates a surplus in the initial potential that pays the entire deletion charge.
	Enumerating small labeled prefixes ensures that such a starting solution is included.
	
	Neither the good prefix nor the ideal weights are known to the algorithm.
	We enumerate small labeled prefixes and finite guesses of the coefficients defining the weights.
	Each guess is used until a remaining budget drops by a prescribed fraction; a fixed stopping threshold bounds the number of these phases.
	At each phase boundary, all guesses inherit the same realized state, including previously sampled labels, and use fresh randomness thereafter.
	This lets the proof follow accurate guesses using only past information.
	The best candidate therefore inherits the comparison guarantee, with an arbitrarily small approximation loss.
	
	\subsection{Related work}
	\paragraph{Foundations and $k$-submodular maximization.}
	Huber and Kolmogorov~\cite{huberkolmogorov2012} introduced the modern $k$-submodular framework.
	Ward and \v{Z}ivn\'y~\cite{wardzivny2016} established its marginal characterization and foundational approximation results.
	For unconstrained maximization, Iwata, Tanigawa, and Yoshida~\cite{iwata2016} gave a randomized $1/2$-approximation for non-monotone objectives and a $k/(2k-1)$-approximation for monotone objectives, together with the monotone oracle lower bound $(k+1)/(2k)$.
	Oshima~\cite{oshima2021} improved the unconstrained non-monotone ratio for fixed $k\ge3$.
	Under support constraints, Ohsaka and Yoshida~\cite{ohsakayoshida2015} and Sakaue~\cite{sakaue2017} established monotone $1/2$ guarantees for size and matroid constraints, respectively; Sun, Liu, and Li~\cite{sun2022} obtained $1/3$ for the non-monotone matroid problem.
	
	\paragraph{Single and multiple knapsacks.}
	Table~\ref{tab:knapsacks} lists selected guarantees.
	For a single knapsack, Tang, Wang, and Chan~\cite{tang2022} gave $\tfrac12(1-e^{-1})$ for monotone objectives.
	Xiao et al.~\cite{xiao2026} improved the monotone guarantee to $\tfrac12(1-e^{-2})$ and obtained combinatorial $\tfrac13(1-e^{-3})$ and continuous $1/3-\varepsilon$ guarantees for non-monotone objectives.
	Ha, Pham, and Tran~\cite{ha2024} developed faster deterministic monotone algorithms, including a $1/3-\varepsilon$ guarantee.
	Wang~\cite{wang2025} showed that one initial guess suffices for $1/2$ in the monotone case and $1/3$ in the non-monotone case.
	Chen et al.~\cite{chen2026} introduced the proportional top-2 analysis, obtaining $\sqrt2-1$ for non-monotone matroid and knapsack constraints; their revised single-budget Bellman analysis gives $1/2$ with $O(n^3k^2)$ queries.
	Our algorithm retains their labeling rule and extends that Bellman argument to several budgets.
	
	For multiple knapsacks, Gong, Liu, and Fang~\cite{gong2026} gave combinatorial streaming ratios $1/(2+2d)-\varepsilon$ and $1/(3+2d)-\varepsilon$ for monotone and non-monotone objectives.
	Spaeh, Ene, and Nguyen~\cite{spaeh2025} study single-pass streaming and online algorithms under cardinality and knapsack constraints.
	Zhou, Huang, and Wang~\cite{zhou2025} developed a $k$-multilinear extension and rounding framework, obtaining $1/2-\varepsilon$ and $1/3-\varepsilon$, respectively, for a fixed number of knapsacks.
	Aggarwal~\cite{aggarwal2026} improved the non-monotone continuous guarantees to the constants in Table~\ref{tab:knapsacks}, with arbitrarily small loss from feasible rounding under fixed multiple budgets.
	For broader support families, Yu et al.~\cite{yu2023} studied matroid--knapsack intersections, while Zhang et al.~\cite{zhang2026} and Ye et al.~\cite{ye2026} gave constraint-dependent guarantees for $p$-systems with multiple knapsacks and related approximate models.
	
	\paragraph{Relation to earlier analyses.}
	Our proof uses the gain--loss comparison familiar from randomized double greedy~\cite{buchbinder2015}.
	Bellman's dynamic-programming principle~\cite{bellman1952} supplies the general continuation viewpoint; the specific resource coefficient and top-2 inequalities used here come from Chen et al.~\cite{chen2026}.
	
	\paragraph{Ordinary submodular maximization.}
	For ordinary monotone submodular maximization, Sviridenko~\cite{sviridenko2004} obtained the optimal single-knapsack factor $1-1/e$ through enumeration and greedy selection.
	For non-monotone objectives, Lee et al.~\cite{lee2010} gave constant approximations under matroid or a fixed number of knapsack constraints.
	Measured continuous greedy~\cite{feldman2011}, combined with knapsack rounding~\cite{kulik2013}, yields $1/e-\varepsilon$ under a fixed number of budgets.
	The nonsymmetric method of Buchbinder and Feldman~\cite{buchbinderfeldman2019} and their later DR-submodular analysis~\cite{buchbinderfeldman2024} improved the general guarantees to $0.385-\varepsilon$ and $0.401-\varepsilon$.
	The latter is the strongest general randomized guarantee identified in our review.
	Feldman and Kuhnle's October 2026 preprint~\cite{feldmankuhnle2026} gives the recent oracle upper bound $8/17$, even for cardinality constraints and hence also for knapsacks.
	
	A separate line improves deterministic approximation and query efficiency.
	Sun et al.~\cite{sun2024} gave a deterministic $1/4$ guarantee for a single knapsack.
	The extended multilinear method of Buchbinder and Feldman~\cite{buchbinderfeldman2025} underlies subsequent deterministic developments; Chen et al.~\cite{chensubmodular2026} obtain $0.367-\varepsilon$ for a single knapsack.
	Pham~\cite{pham2025} developed linear-query deterministic algorithms, and Liu and Zhang's September 2026 preprint~\cite{liuzhang2026} achieves $1/4-\varepsilon$ with $O(n\log^2(1/\varepsilon)/\varepsilon^2)$ queries.
	The unconstrained $1/2$ oracle bound of Feige, Mirrokni, and Vondr\'ak~\cite{feige2011} also yields the bisubmodular lower bound~\cite{wardzivny2016}.
	
	\paragraph{Organization.}
	Section~\ref{sec:prelim} defines the labeling operations, the top-2 rule, and the comparison loss.
	Section~\ref{sec:core} develops the resource weights and potential, then bounds deletions and explains their compensation by a prefix.
	Sections~\ref{sec:algorithm} and~\ref{sec:implementation} give the finite implementation and its approximation and query bounds.
	The appendices contain the local algebra and parameter estimates.
	
	
	\section{Preliminaries}\label{sec:prelim}
	Let $V$ be a ground set of $n$ elements and let $[k]=\{1,\ldots,k\}$, where $k\ge2$.
	We use the notation of Chen et al.~\cite{chen2026}.
	The domain $(k+1)^V$ consists of tuples $\mathbf{x}=(X_1,\ldots,X_k)$ of pairwise disjoint subsets of $V$.
	Equivalently, $\mathbf{x}(e)\in\{0\}\cup[k]$ is the label of $e$, with label zero meaning unselected.
	Write $\supp(\mathbf{x})=\bigcup_i X_i$ and $\mathbf{x}\preceq\mathbf{y}$ if $X_i\subseteq Y_i$ for every $i$.
	The empty labeling is $\zero$.
	For $\mathbf{x},\mathbf{y}\in(k+1)^V$, define
	\[
	\mathbf{x}\sqcap\mathbf{y}=(X_1\cap Y_1,\ldots,X_k\cap Y_k),\qquad
	\mathbf{x}\sqcup\mathbf{y}=(Z_1,\ldots,Z_k),
	\quad Z_i=(X_i\cup Y_i)\setminus\bigcup_{\ell\ne i}(X_\ell\cup Y_\ell).
	\]
	A function $f:(k+1)^V\to\R$ is $k$-submodular if, for all $\mathbf{x},\mathbf{y}\in(k+1)^V$,
	\[
	f(\mathbf{x})+f(\mathbf{y})\ge f(\mathbf{x}\sqcap\mathbf{y})+f(\mathbf{x}\sqcup\mathbf{y}).
	\]
	Throughout the paper, $f$ is nonnegative and is accessed through an exact value oracle.
	For an unselected element $e$, let $\mathbf{x}+i\one_e$ denote its assignment to label $i$, and put
	\[
	\Delta_{e,i}f(\mathbf{x})=f(\mathbf{x}+i\one_e)-f(\mathbf{x}).
	\]
	The marginal characterization of Ward and \v{Z}ivn\'y~\cite{wardzivny2016} states that $k$-submodularity is equivalent to orthant submodularity and pairwise monotonicity:
	\begin{align}
		\Delta_{e,i}f(\mathbf{x})&\ge\Delta_{e,i}f(\mathbf{y})
		&&(\mathbf{x}\preceq\mathbf{y},\ e\notin\supp(\mathbf{y})),\label{eq:orthant}\\
		\Delta_{e,i}f(\mathbf{x})+\Delta_{e,\ell}f(\mathbf{x})&\ge0
		&&(e\notin\supp(\mathbf{x}),\ i\ne\ell).\label{eq:pairwise}
	\end{align}
	
	There are $d$ knapsack constraints with budgets $B_j\ge0$. Element $e$ has label-independent costs $c_{j,e}\ge0$ and a feasible labeling satisfies
	\[
	\sum_{e\in\supp(\mathbf{x})}c_{j,e}\le B_j\qquad(j\in[d]).
	\]
	Write $\cost_e=(c_{1,e},\ldots,c_{d,e})^{\mathsf T}$ for its cost vector.
	Zero-budget rows first exclude elements of positive cost in that row; the row is then replaced by a redundant unit-budget zero-cost row.
	All positive budgets are normalized to one, and individually infeasible elements are removed.
	These operations preserve the feasible domain and leave exactly $d$ rows.
	Let $\mathbf{o}$ be an optimal feasible labeling and let $\OPT=f(\mathbf{o})$.
	Vector inequalities are componentwise; $\pos{a}=\max\{a,0\}$.
	Fixed deterministic tie-breaking is used throughout.
	Table~\ref{tab:notation} collects the main notation, including the comparison quantities introduced in Section~\ref{sec:core}.
	\begin{table}[!htb]
		\centering
		\small
		\caption{Main notation. Continuation quantities refer to the problem remaining after the fixed prefix $\mathbf y$; its remaining budgets are normalized to one.}
		\label{tab:notation}
		\renewcommand{\arraystretch}{1.13}
		\begin{tabularx}{\textwidth}{@{}p{.22\textwidth}X@{}}
			\toprule
			Symbol & Meaning\\
			\midrule
			\multicolumn{2}{@{}l}{\emph{Input and quantities maintained by the algorithm}}\\[2pt]
			$V,n,k,d$ & Ground set, number of elements, number of labels, and number of knapsack constraints.\\
			$\mathbf y,\ f_{\mathbf y}$ & Fixed labeled prefix and conditioned objective $f_{\mathbf y}(\mathbf x)=f(\mathbf y\sqcup\mathbf x)$, retaining the baseline value $f(\mathbf y)$.\\
			$\mathbf s,\ U$ & Current continuation labeling and its unselected candidates that have not been permanently removed.\\
			$c_{j,e},\ \cost_e,\ r_j$ & Cost of $e$ in row $j$, its cost vector, and the resource still available: $r_j=1-\sum_{e\in\supp(\mathbf s)}c_{j,e}$.\\
			$\mathbf g_e,\mathbf q_e,G_e$ & Marginal-gain vector at $\mathbf s$, proportional top-2 label distribution, and expected gain $G_e=\mathbf q_e^{\mathsf T}\mathbf g_e$.\\
			$\widehat\lambda_j,\widehat p_j$ & Guessed coefficient and resulting resource weight, fixed within a phase; the greedy density is $G_e/\sum_j\widehat p_jc_{j,e}$.\\
			$H,\kappa,h$ & Maximum enumerated prefix size, single-element safety fraction $c_{j,e}\le\kappa r_j$, and phase width.\\
			\midrule
			\multicolumn{2}{@{}l}{\emph{Quantities for the comparison analysis}}\\[2pt]
			$\mathbf o,\ \OPT$ & Fixed optimal labeling and its value $f(\mathbf o)$.\\
			$\mathbf z$ & Comparison labeling: the algorithm's current labels, together with the original labels of unselected, undeleted residual optimum elements.\\
			$\theta_e,\ m_j$ & Indicator that $e$ belongs to the residual optimum, and comparison resource with slack $t>0$: $m_j=1+t-\sum_{e\in\supp(\mathbf s)}\theta_ec_{j,e}$. Only selected optimum elements reduce $m_j$.\\
			$\lambda_j,\ \Lambda$ & Row coefficient $\lambda_j=r_j/(Qm_j+br_j)$ and common comparison coefficient, approximating the smallest $\lambda_j$.\\
			$p_j,\ R,\ K$ & Ideal resource weight, multiplicative error bound along the comparison path, and the compensation factor $K=R^2$.\\
			$D,\ \Psi$ & Comparison gap $D=f_{\mathbf y}(\mathbf z)-f_{\mathbf y}(\mathbf s)$ and potential $\Psi=Kf_{\mathbf y}(\mathbf s)+\Lambda\pos D$.\\
			$L$ & Largest positive marginal of a residual optimum element, in its original label, at the prefix $\mathbf y$.\\
			\bottomrule
		\end{tabularx}
	\end{table}
	

	\subsection{The proportional top-2 rule}
	Given a marginal vector $\mathbf{g}=(g_1,\ldots,g_k)$ satisfying pairwise monotonicity, let $\alpha,\beta$ index its two largest entries, with $g_\alpha\ge g_\beta$.
	Put $y_1=\pos{g_\alpha}$ and $y_2=\pos{g_\beta}$.
	If $y_1=0$, let $\mathbf{q}$ be concentrated on $\alpha$ and set $G=0$.
	Otherwise define
	\[
	q_\alpha=\frac{y_1}{y_1+y_2},\qquad
	q_\beta=\frac{y_2}{y_1+y_2},\qquad
	q_i=0\ (i\notin\{\alpha,\beta\}),\qquad
	G=\mathbf{q}^{\mathsf T}\mathbf{g}=\frac{y_1^2+y_2^2}{y_1+y_2}.
	\]
	We write $(\mathbf{q},G)=\Top(\mathbf{g})$.
	Every sampled label has nonnegative marginal, so $G$ is the exact expected gain.
	The constants used below are the constants from Chen et al.~\cite{chen2026}:
	\[
	A=\frac{\sqrt2-1}{2},\qquad Q=1+A=\frac{1+\sqrt2}{2},\qquad b=2-Q.
	\]
	The first two bounds below control changes in the comparison labeling defined next; the third relates an individual marginal to the expected gain $G$.
	
	\begin{lemma}[top-2 certificate]\label{lem:top2}
		Let $(\mathbf{q},G)=\Top(\mathbf{g})$. For every pairwise-monotone vector $\mathbf{v}\le\mathbf{g}$ and every label $i \in [k]$,
		\[
		v_i-\mathbf{q}^{\mathsf T}\mathbf{v}\le G,\qquad
		-\mathbf{q}^{\mathsf T}\mathbf{v}\le AG,\qquad
		\max_i g_i\le QG.
		\]
	\end{lemma}
	Appendix~\ref{app:top2} gives a direct proof.
	
	\section{A common Bellman potential}\label{sec:core}
	We first analyze a continuation from a fixed part $\mathbf{y}\preceq\mathbf{o}$ of the optimum.
	The continuation uses approximately correct resource weights; Section~\ref{sec:algorithm} shows how to include such a continuation without knowing the optimum.
	Fix the elements and labels of $\mathbf{y}$, and define
	\[
	f_{\mathbf{y}}(\mathbf{x})=f(\mathbf{y}\sqcup\mathbf{x})
	\qquad\bigl(\supp(\mathbf{x})\subseteq V\setminus\supp(\mathbf{y})\bigr).
	\]
	Thus $f_{\mathbf{y}}(\zero)=f(\mathbf{y})$.
	On this residual domain, $f_{\mathbf{y}}$ remains nonnegative and inherits orthant submodularity and pairwise monotonicity from $f$.
	After normalizing the residual budgets, again write $c_{j,e}$ for the remaining costs; every row of the residual optimum has cost at most one.
	If a residual budget is zero, perform the same zero-budget preprocessing as in Section~\ref{sec:prelim}.
	Because $\mathbf{y}\preceq\mathbf{o}$, no residual optimum element is excluded by this preprocessing.
	
	A continuation selects whole elements and samples their labels using the proportional top-2 rule.
	Let $\mathbf{s}$ be its current labeling and $r_j=1-\sum_{e\in\supp(\mathbf{s})}c_{j,e}$ its physical remaining resource.
	Write $\res=(r_1,\ldots,r_d)^{\mathsf T}$.
	Before selecting a positive-cost element, it permanently deletes each unselected element for which $c_{j,e}>\kappa r_j$ in some row, where $0<\kappa<1$.
	Call these elements unsafe: although they may still fit, their cost is too large relative to the remaining resource for the analysis below.
	Deletion removes an element from future consideration, not from the actual labeling.
	The surviving unselected elements form the candidate set $U$.
	For each candidate, compute
	\[
	\mathbf{g}_e=(\Delta_{e,i}f_{\mathbf{y}}(\mathbf{s}))_{i\in[k]},\qquad
	(\mathbf{q}_e,G_e)=\Top(\mathbf{g}_e).
	\]
	Given positive resource weights $\hprices$, held fixed for a phase, the next positive-cost element maximizes
	\begin{equation}\label{eq:density}
		\frac{G_e}{\hprices^{\mathsf T}\cost_e}.
	\end{equation}
	Zero-cost elements are selected first in a fixed order.
	The denominator $\hprices^{\mathsf T}\cost_e$ is the weighted cost of $e$: $c_{j,e}$ is its fixed resource consumption, whereas $\widehat p_j$ is the weight assigned to that resource in the selection rule.
	The safety test implies $r'_j\ge(1-\kappa)r_j>0$ after every positive-cost selection.
	
	\subsection{The comparison labeling and its loss}
	Following the hybrid-solution analysis of Chen et al.~\cite{chen2026}, maintain a comparison labeling $\mathbf{z}$.
	It contains the actual labels on $\supp(\mathbf{s})$ and the original optimum labels on all unselected, undeleted residual optimum elements; all other elements have label zero.
	Thus $\mathbf{s}\preceq\mathbf{z}$, although $\mathbf{z}$ need not be feasible.
	Initially $\mathbf{s}=\zero$ and $\mathbf{z}$ is the residual optimum.
	As the continuation proceeds, $\mathbf{z}$ changes in three ways:
	\begin{itemize}
		\item Selecting an optimum element replaces its optimum label in $\mathbf{z}$ by its sampled label. This is a \emph{replacement}, including the case when the labels agree.
		\item Selecting an element outside the optimum adds its sampled label to $\mathbf{z}$. This is a \emph{pure addition}.
		\item Deleting an unselected optimum element resets its label in $\mathbf{z}$ to zero. Deleting any other element leaves $\mathbf{z}$ unchanged.
	\end{itemize}
	Write $\mathbf{z}-e$ for resetting the label of $e$ to zero.
	For a selection, let $\mathcal F$ denote the history before sampling its label and let primes denote the resulting state.
	The \emph{comparison loss}, called hybrid loss in~\cite{chen2026}, is the signed quantity
	\[
	f_{\mathbf{y}}(\mathbf{z})-\E[f_{\mathbf{y}}(\mathbf{z}')\mid\mathcal F].
	\]
	For a replacement of $e$, set $v_i=\Delta_{e,i}f_{\mathbf{y}}(\mathbf{z}-e)$; its comparison loss is $v_{\mathbf{o}(e)}-\mathbf{q}_e^{\mathsf T}\mathbf{v}\le G_e$.
	For a pure addition, set $v_i=\Delta_{e,i}f_{\mathbf{y}}(\mathbf{z})$; its comparison loss is $-\mathbf{q}_e^{\mathsf T}\mathbf{v}\le AG_e$.
	In both cases $\mathbf{v}\le\mathbf{g}_e$ by~\eqref{eq:orthant}, and $\mathbf{v}$ is pairwise monotone by~\eqref{eq:pairwise}, so Lemma~\ref{lem:top2} applies.
	The comparison loss may be negative because $\mathbf{z}$ can gain value.
	Deletion has no label sampling and will be bounded separately.
	
	\subsection{From one resource to a common coefficient}
	The actual value and the comparison value differ by
	\[
	D=f_{\mathbf{y}}(\mathbf{z})-f_{\mathbf{y}}(\mathbf{s}).
	\]
	To control this gap under several budgets, we need one coefficient that respects every resource row.
	Fix $0<t<1/2$ and $\tau>0$; their sizes will be chosen from $\varepsilon$.
	Let $\theta_e=1$ if $e$ belongs to the residual optimum and $\theta_e=0$ otherwise.
	These indicators are fixed for the entire analysis, including after an element is deleted.
	Introduce
	\[
	m_j=1+t-\sum_{e\in\supp(\mathbf{s})}\theta_ec_{j,e},\qquad
	\lambda_j=\frac{r_j}{Qm_j+br_j}.
	\]
	Write $\ghost=(m_1,\ldots,m_d)^{\mathsf T}$ and $\lvec=(\lambda_1,\ldots,\lambda_d)^{\mathsf T}$.
	Unlike the actual remainder $r_j$, the comparison resource $m_j$ decreases only when an optimum element is selected.
	In particular, deleting an optimum element does not reduce $m_j$; its cost remains available in the comparison accounting.
	The residual optimum has total cost at most one. Consequently,
	\[
	m_j\ge t+\sum_{e\in\supp(\mathbf{o})\setminus\supp(\mathbf{y}\sqcup\mathbf{s})}c_{j,e},
	\qquad m_j-r_j=t+\sum_{e\in\supp(\mathbf{s})}(1-\theta_e)c_{j,e}\ge t.
	\]
	The slack $t$ keeps $m_j$ positive, and $0<\lambda_j\le1/2$.
	Each $\lambda_j$ is the single-resource Bellman coefficient for row $j$.
	To obtain one coefficient below every $\lambda_j$, take a smooth approximation of their minimum, subtract $t$, and truncate at zero:
	\begin{equation}\label{eq:envelope}
		\Lambda=\pos{-\tau\log\sum_{j=1}^d e^{-\lambda_j/\tau}-t}.
	\end{equation}
	Then $0\le\Lambda\le\min_j\lambda_j$.
	The smoothing parameter $\tau$ permits differentiation, while the shift by $t$ lets the comparison term vanish before a resource is exhausted.
	Initially, $r_j=1$ and $m_j=1+t$ for every $j$.
	Since $b=2-Q$, all single-budget coefficients have the same initial value:
	\[
	\lambda_j=\frac{1}{Q(1+t)+b}=\frac{1}{2+Qt}.
	\]
	Substituting these equal coefficients into~\eqref{eq:envelope} gives
	\begin{equation}\label{eq:initial}
		\Lambda_0
		=\pos{\frac{1}{2+Qt}-t-\tau\log d}
		=\frac{1}{2+Qt}-t-\tau\log d,
	\end{equation}
	where the parameter choices below ensure that the expression inside
	the positive part is strictly positive.
	When $\Lambda>0$, define
	\begin{equation}\label{eq:prices}
		w_j=\frac{e^{-\lambda_j/\tau}}{\sum_{\ell=1}^d e^{-\lambda_\ell/\tau}},\qquad
		\phi(\lambda)=\frac{\lambda(1-b\lambda)}{1-Q\lambda},\qquad
		p_j=w_j\frac{\phi(\lambda_j)}{r_j}.
	\end{equation}
	Here $w_j$ is the derivative coefficient of row $j$ in the smooth minimum, and $\phi$ converts its Bellman coefficient into a resource weight.
	The resulting $p_j$ is an ideal weight for the analysis, distinct from the frozen weight $\widehat p_j$ used by the algorithm.
	It depends on the unknown optimum through $m_j$; the algorithm will approximate it by guessing the coefficients $\lambda_j$.
	
	To bound the change in the resource coefficient caused by selecting $e$,
	consider the straight line between the resource states before and after
	the selection. For $u\in[0,1]$, define
	\[
	\begin{aligned}
		r_j(u)&=r_j(0)-u c_{j,e},&
		m_j(u)&=m_j(0)-u\theta_e c_{j,e},\\
		\lambda_j(u)&=\frac{r_j(u)}{Qm_j(u)+br_j(u)}.
	\end{aligned}
	\]
	The endpoints $u=0$ and $u=1$ correspond to the states before and after
	selecting the whole element.
	The safety test and the initial resource invariants give, throughout this path,
	\[
	\begin{aligned}
		r_j(u)&\ge(1-\kappa)r_j(0)>0,\\
		m_j(u)-r_j(u)&=m_j(0)-r_j(0)+(1-\theta_e)u c_{j,e}\ge t.
	\end{aligned}
	\]
	Hence $0<\lambda_j(u)\le1/2$ throughout the update.
	Writing a dot for differentiation with respect to $u$, we obtain
	\[
	-\dot\lambda_j(u)
	=\frac{c_{j,e}}{r_j(u)}\lambda_j(u)
	\bigl[1-(b+Q\theta_e)\lambda_j(u)\bigr].
	\]
	This path is used only to analyze the resource change; objective values
	are evaluated only at discrete labelings.
	
	Since $0<\lambda_j(u)\le1/2$ and $0<b+Q\theta_e\le2$,
	the factor $1-(b+Q\theta_e)\lambda_j(u)$ lies in $[0,1]$.
	Consequently,
	\begin{equation}\label{eq:relative}
		\begin{aligned}
			0\le-\frac{\mathrm d}{\mathrm du}\log\lambda_j(u)
			&=\frac{c_{j,e}}{r_j(u)}
			\bigl[1-(b+Q\theta_e)\lambda_j(u)\bigr]\\
			&\le\frac{c_{j,e}}{r_j(u)}
			=-\frac{\mathrm d}{\mathrm du}\log r_j(u).
		\end{aligned}
	\end{equation}
	Thus every $\lambda_j(u)$ is nonincreasing.
	The smooth minimum and positive-part operation preserve this
	monotonicity, so $\Lambda(u)$ is also nonincreasing.
	
	For the resource inequalities below, define
	$
	\beta_e=Q+b\theta_e.
	$
	Since $Q+b=2$, we have $\beta_e=Q$ when $e$ is outside the
	fixed optimum and $\beta_e=2$ when it belongs to the optimum.
	
	\begin{lemma}[Resource-weight inequalities]\label{lem:resource}
		While $\Lambda>0$, the resource weights in~\eqref{eq:prices} satisfy
		\[
		Q\sum_{j=1}^d p_jm_j\le1,\qquad
		-\dot\Lambda\le(1-\beta_e\Lambda)\prices^{\mathsf T}\cost_e.
		\]
	\end{lemma}
	The first inequality bounds the total weighted comparison resource.
	The second bounds the decrease in the common coefficient caused by selecting $e$.
	They will use the same weighted cost in the density and potential calculations, allowing that cost to cancel.
	\begin{proof}
		Since $Qm_j/r_j=1/\lambda_j-b$,
		\[
		\frac{Qp_jm_j}{w_j}=\frac{(1-b\lambda_j)^2}{1-Q\lambda_j}.
		\]
		The denominator is positive since $Q<2$ and $\lambda_j\le1/2$. The numerator minus the denominator is
		\[
		\lambda_j(3Q-4+b^2\lambda_j)
		\le\frac{\lambda_j}{2}(Q^2+2Q-4)<0,
		\]
		where $Q^2+2Q-4=(6\sqrt2-9)/4<0$.
		Summing $Qp_jm_j\le w_j$ proves the first inequality.
		For $\theta_e=0$,
		\[
		w_j\frac{\lambda_j(1-b\lambda_j)}{r_j}
		=(1-Q\lambda_j)p_j\le(1-Q\Lambda)p_j.
		\]
		For $\theta_e=1$, $Q\ge b$ gives
		\[
		w_j\frac{\lambda_j(1-2\lambda_j)}{r_j}
		=(1-2\lambda_j)p_j\frac{1-Q\lambda_j}{1-b\lambda_j}
		\le(1-2\Lambda)p_j.
		\]
		Differentiating the smooth minimum yields $\dot\Lambda=\sum_jw_j\dot\lambda_j$ whenever $\Lambda>0$.
		Multiplication by $c_{j,e}$ and summation prove the second inequality.
		When $\Lambda=0$, it remains zero along the update and its derivative is zero almost everywhere.
	\end{proof}
	
	\subsection{The whole-element potential}
	The selection rule~\eqref{eq:density} uses frozen weights
	$\widehat p_j$, whereas the resource inequalities involve
	the ideal weights $p_j$, which change as resources are consumed.
	For the selected element $e$, assume that these weights remain
	within a factor $R\ge1$ throughout its resource update:
	\begin{equation}\label{eq:distortion}
		R^{-1}\widehat p_j\le p_j(u)\le R\widehat p_j
		\qquad
		(j\in[d],\ u\in[0,1]\text{ with }\Lambda(u)>0).
	\end{equation}
	We will establish this approximation for the implemented weights
	in Lemma~\ref{lem:price}.
	
	To compare the objective gain with the loss in the comparison
	term, define
	\[
	\Psi=Kf_{\mathbf y}(\mathbf s)+\Lambda\pos{D},
	\qquad K=R^2.
	\]
	The first term is the current objective value scaled by $K$.
	The second weights the positive comparison gap by $\Lambda$;
	it vanishes when the comparison labeling has no larger value
	than the current labeling.
	
	The proof below shows that the conditional expected decrease in
	the comparison term is at most $R^2G_e$.
	Since the conditional expected gain is $G_e$,
	the choice $K=R^2$ makes the expected increase in the first
	term sufficient to cover this loss.
	Thus a whole-element selection does not decrease $\Psi$
	in conditional expectation.
	
	\begin{lemma}[Potential change under one selection]\label{lem:integer}
		Let $\mathcal F$ be the history immediately before sampling the
		selected element's label, after any unsafe deletions.
		For a positive-cost selection, assume~\eqref{eq:distortion}.
		Let $\Psi'$ denote the potential immediately after the selection
		and before any further deletions. Then
		\[
		\E[\Psi'\mid\mathcal F]\ge\Psi.
		\]
		The same conclusion holds for zero-cost selections without
		a weight approximation condition.
	\end{lemma}
	
	\begin{proof}
		Let $e$ be the selected element.
		Conditional on $\mathcal F$, both $e$ and its expected gain $G_e$
		are fixed.
		Since costs are label-independent, the resource update and hence
		the new coefficient $\Lambda'$ are also fixed before the label
		is sampled.
		
		We first bound the change in the comparison gap.
		If $\theta_e=1$, then $e$ still carries its optimum label in
		$\mathbf z$.
		Set
		\[
		v_i=\Delta_{e,i}f_{\mathbf y}(\mathbf z-e).
		\]
		By~\eqref{eq:orthant}--\eqref{eq:pairwise},
		$\mathbf v\le\mathbf g_e$ and $\mathbf v$ is pairwise monotone.
		The replacement bound in Lemma~\ref{lem:top2} therefore gives
		\[
		\mathbf q_e^{\mathsf T}\mathbf v-v_{\mathbf o(e)}
		\ge -G_e.
		\]
		If $\theta_e=0$, set
		$v_i=\Delta_{e,i}f_{\mathbf y}(\mathbf z)$ instead.
		The same properties hold, and the pure-addition bound gives
		\[
		\mathbf q_e^{\mathsf T}\mathbf v\ge-AG_e.
		\]
		These expressions are the expected changes in the comparison
		value in the two cases.
		Subtracting the expected increase $G_e$ in the current value,
		and recalling that $\beta_e=2$ when $\theta_e=1$ and
		$\beta_e=Q=1+A$ otherwise, yields
		\begin{equation}\label{eq:gapstep}
			\E[D'\mid\mathcal F]\ge D-\beta_eG_e,
			\qquad
			\E[f_{\mathbf y}(\mathbf s')-f_{\mathbf y}(\mathbf s)
			\mid\mathcal F]=G_e.
		\end{equation}
		
		Now suppose that $D>0$, $\Lambda>0$, and $\cost_e\ne\zero$.
		Write
		\[
		a_e=\hprices^{\mathsf T}\cost_e>0
		\]
		for the weighted cost of the selected element.
		Every element in
		$\supp(\mathbf z)\setminus\supp(\mathbf s)$ remains a candidate.
		Its cost vector is nonzero because all zero-cost elements were
		processed first.
		Adding these elements in their optimum labels and using
		orthant submodularity gives
		\begin{equation}\label{eq:marginalgap}
			D\le
			\sum_{a\in\supp(\mathbf z)\setminus\supp(\mathbf s)}
			\Delta_{a,\mathbf o(a)}f_{\mathbf y}(\mathbf s)
			\le
			Q\sum_{a\in\supp(\mathbf z)\setminus\supp(\mathbf s)}G_a,
		\end{equation}
		where the second inequality follows from Lemma~\ref{lem:top2}.
		This bound does not depend on the selected element or the resource weights.
		Density maximality bounds each $G_a$ by
		$(G_e/a_e)\hprices^{\mathsf T}\cost_a$.
		Moreover, the total cost vector of these comparison elements
		is bounded componentwise by $\ghost$.
		Consequently,
		\begin{equation}\label{eq:densitygap}
			D\le Q\frac{G_e}{a_e}\hprices^{\mathsf T}\ghost
			\le R\frac{G_e}{a_e}.
		\end{equation}
		The last inequality uses~\eqref{eq:distortion} at the initial
		state and $Q\prices^{\mathsf T}\ghost\le1$ from
		Lemma~\ref{lem:resource}.
		
		Since $\Lambda(u)$ is absolutely continuous, we may integrate
		the second inequality of Lemma~\ref{lem:resource} along the
		resource update, only while $\Lambda(u)>0$.
		Since
		$\prices(u)^{\mathsf T}\cost_e\le Ra_e$ and
		$\Lambda(u)\ge\Lambda'$, we obtain
		\[
		\Lambda-\Lambda'
		\le Ra_e(1-\beta_e\Lambda').
		\]
		Here the interpolation interval has length one; if $\Lambda$
		reaches zero earlier, it remains zero thereafter.
		Multiplying this bound by~\eqref{eq:densitygap} cancels the
		weighted cost $a_e$ and gives
		\[
		(\Lambda-\Lambda')D
		\le R^2(1-\beta_e\Lambda')G_e
		=K(1-\beta_e\Lambda')G_e.
		\]
		Thus the two uses of the weight approximation produce the
		factor $R^2$.
		
		Conditional Jensen and~\eqref{eq:gapstep} give
		\[
		\E[\pos{D'}\mid\mathcal F]
		\ge \pos{\E[D'\mid\mathcal F]}
		\ge D-\beta_eG_e.
		\]
		Since $\Lambda'$ is fixed conditional on $\mathcal F$,
		the expected decrease in the comparison term satisfies
		\begin{align*}
			\Lambda\pos{D}
			-\E[\Lambda'\pos{D'}\mid\mathcal F]
			&\le(\Lambda-\Lambda')D+\beta_e\Lambda'G_e\\
			&\le K(1-\beta_e\Lambda')G_e+\beta_e\Lambda'G_e\\
			&=\bigl[K-(K-1)\beta_e\Lambda'\bigr]G_e\\
			&\le KG_e=R^2G_e.
		\end{align*}
		The first line accounts for both the decrease in the coefficient
		and the change in the comparison gap.
		The objective term has conditional expected increase $KG_e$,
		so
		\[
		\E[\Psi'-\Psi\mid\mathcal F]
		\ge (K-1)\beta_e\Lambda'G_e
		\ge0.
		\]
		
		It remains to handle the other cases.
		If $D\le0$, the old comparison term is zero, the new comparison
		term is nonnegative, and every actual gain is nonnegative.
		If $\Lambda=0$, it remains zero, so the nonnegative objective
		gain again suffices.
		Finally, a zero-cost selection leaves all resources unchanged,
		and hence $\Lambda'=\Lambda$.
		When $D>0$, the same Jensen bound gives
		\[
		\E[\Psi'-\Psi\mid\mathcal F]
		\ge (K-\beta_e\Lambda)G_e
		\ge0,
		\]
		because $K\ge1$, $\beta_e\le2$, and $\Lambda\le1/2$.
	\end{proof}
	
	\subsection{Paying for unsafe optimum elements}
	The density argument in~\eqref{eq:densitygap} compares the selected element with every remaining element of $\mathbf{z}$.
	These elements must therefore remain in $U$.
	Whenever an optimum element becomes unsafe, removing it from both $U$ and $\mathbf{z}$ restores this property but may lower the potential.
	Let
	\[
	L=\max_{e\in\supp(\mathbf{o})\setminus\supp(\mathbf{y})}\pos{\Delta_{e,\mathbf{o}(e)}f(\mathbf{y})},
	\]
	with $L=0$ if the set is empty.
	This single bound, measured at the fixed prefix, controls every later deletion.
	Removing an optimum element $e$ from $\mathbf{z}$ can decrease $f_{\mathbf{y}}(\mathbf{z})$ by at most $L$:
	its removal loss is $\Delta_{e,\mathbf{o}(e)}f(\mathbf{y}\sqcup(\mathbf{z}-e))\le L$ by~\eqref{eq:orthant}.
	The potential decrease is at most
	\[
	\Lambda\pos{f_{\mathbf{y}}(\mathbf{z})-f_{\mathbf{y}}(\mathbf{z}-e)}\le\Lambda L,
	\]
	because the positive-part function is nondecreasing and $1$-Lipschitz, and deletion leaves $\Lambda$ unchanged.
	Deleting an element outside the residual optimum leaves the comparison labeling and potential unchanged.
	
	\begin{lemma}[Total deletion charge]\label{lem:deletion}
		On every realized continuation for $\mathbf{y}\preceq\mathbf{o}$,
		\[
		\sum_{e\text{ deleted from the residual optimum}}\Lambda_e\le\frac{d}{\kappa Q},
		\]
		where $\Lambda_e$ is the coefficient at the deletion of $e$.
	\end{lemma}
	\begin{proof}
		Assign each deleted optimum element to one row $j$ witnessing $c_{j,e}>\kappa r_j$.
		Let $T_j$ be the total row-$j$ cost of all optimum elements eventually assigned to $j$.
		They are never selected. The lower bound on $m_j$ in Section~\ref{sec:core} therefore gives $m_j\ge T_j$ at every deletion time.
		If $T_j>0$, then for each assigned element
		\[
		\Lambda_e\le\lambda_j=\frac{r_j}{Qm_j+br_j}
		\le\frac{r_j}{Qm_j}
		<\frac{c_{j,e}}{\kappa Qm_j}
		\le\frac{c_{j,e}}{\kappa QT_j}.
		\]
		Summing over that row gives at most $1/(\kappa Q)$.
		If $T_j=0$, no element is assigned to $j$, since a blocking cost is strictly positive.
		Summing the row bounds proves the claimed bound.
		Future assignments appear only in this pathwise analysis; the algorithm does not compute $T_j$.
	\end{proof}
	Hence the total potential decrease from deletions is at most $dL/(\kappa Q)$ on each realization, even if many optimum elements are deleted.
	We next choose a prefix making $L$ small relative to the value already obtained.
	
	\begin{lemma}[An optimum prefix]\label{lem:prefix}
		For every integer $H\ge1$, either a restriction of $\mathbf{o}$ of size at most $H$ has value $\OPT$, or there is a restriction $\mathbf{y}\preceq\mathbf{o}$ of size $H$ satisfying $L\le f(\mathbf{y})/H$.
	\end{lemma}
	\begin{proof}
		Starting from $\zero$, repeatedly choose an unselected optimum element with maximum positive marginal in its original optimum label.
		The chosen marginals are positive and nonincreasing by~\eqref{eq:orthant}.
		If the process stops before $H$ selections, all remaining original-label marginals are nonpositive.
		Adding them in any order and using~\eqref{eq:orthant} yields $f(\mathbf{o})\le f(\mathbf{y})$.
		Feasibility of $\mathbf{y}$ gives the reverse inequality.
		Otherwise the $H$th marginal bounds every remaining original-label marginal from above, and is at most the average of the first $H$ marginals, namely $[f(\mathbf{y})-f(\zero)]/H\le f(\mathbf{y})/H$.
	\end{proof}
	
	\begin{proposition}[Bellman comparison]\label{prop:comparison}
		Suppose $\Lambda_0>0$, weight distortion~\eqref{eq:distortion} holds for every positive-cost selection, and the continuation terminates in a state satisfying $\Lambda\pos{D}=0$.
		For the prefix from Lemma~\ref{lem:prefix}, if
		\begin{equation}\label{eq:prefixpayment}
			H(K-\Lambda_0)\ge\frac{d}{\kappa Q},
		\end{equation}
		then $\E f(\mathbf{y}\sqcup\mathbf{s}_{\rm fin})\ge(\Lambda_0/K)\OPT$.
	\end{proposition}
	\begin{proof}
		If the first alternative of Lemma~\ref{lem:prefix} holds, every continuation has value at least $\OPT$, because all actual gains are nonnegative.
		Otherwise take its size-$H$ restriction.
		Initially $\mathbf{s}=\zero$, $\mathbf{z}$ is the residual optimum, and
		\[
		\Psi_0=Kf(\mathbf{y})+\Lambda_0(\OPT-f(\mathbf{y})).
		\]
		The gap is nonnegative because $\mathbf{y}$ is feasible.
		Equivalently, $\Psi_0=\Lambda_0\OPT+(K-\Lambda_0)f(\mathbf{y})$.
		The second term is the initial surplus available to pay the deletion charge; condition~\eqref{eq:prefixpayment} makes that surplus sufficient when $L\le f(\mathbf{y})/H$.
		Add $L\sum_{e\text{ deleted so far}}\Lambda_e$ to $\Psi$.
		Every deletion leaves this augmented potential nondecreasing, and Lemma~\ref{lem:integer} gives a nonnegative conditional increment at every selection.
		There are at most $n$ selections and $n$ deletions; padding terminated paths with identity steps gives a deterministic finite horizon.
		The tower property and Lemma~\ref{lem:deletion} therefore give
		\begin{align*}
			K\E f(\mathbf{y}\sqcup\mathbf{s}_{\rm fin})
			&\ge\Lambda_0\OPT+(K-\Lambda_0)f(\mathbf{y})-\frac{dL}{\kappa Q}\\
			&\ge\Lambda_0\OPT+\left(K-\Lambda_0-\frac{d}{\kappa QH}\right)f(\mathbf{y})
			\ge\Lambda_0\OPT.
		\end{align*}
		The terminal-state assumption makes the final comparison term zero.
	\end{proof}
	The continuation does not need to test the unknown quantities $\Lambda$ or $D$.
	Lemma~\ref{lem:terminal} will verify this terminal condition using observable stopping rules.
	
	\section{A finite combinatorial algorithm}\label{sec:algorithm}
	The analysis used two objects the algorithm cannot compute: a good optimum prefix and ideal resource weights determined by optimum membership.
	We enumerate the small prefixes, and then enumerate coefficient guesses in phases to approximate the ideal resource weights.
	The tree construction below preserves the sampled labels already obtained before each new guess.
	
	\subsection{Parameters}
	For $0<\varepsilon<1/2$, set
	\[
	t=\frac{\varepsilon}{5/2+Q/4},\qquad
	\tau=\frac{t}{1+\log d},\qquad
	\rho=\frac{Qt^2}{1-2t}.
	\]
	The logarithmic derivative of $\phi$ from~\eqref{eq:prices} satisfies
	\begin{equation}\label{eq:condition}
		(\log\phi)'(\lambda)=\frac1\lambda+
		\frac{Q-b}{(1-b\lambda)(1-Q\lambda)}
		\le\frac1\lambda+c_Q,\qquad
		c_Q=\frac{8(Q-1)}{Qb}<2
		\quad(0<\lambda\le1/2).
	\end{equation}
	Choose
	\begin{equation}\label{eq:parameters}
		h=\frac{t}{2(2/\tau+2/t+c_Q+2)},\quad
		\kappa=1-e^{-h},\quad R=e^{t/2},\quad K=e^t,\quad
		H=\ceil{\frac{d}{\kappa Q(K-\Lambda_0)}}.
	\end{equation}
	Here $\Lambda_0$ is given by~\eqref{eq:initial}.
	The phase width $h$ makes the logarithmic weight error at most $t/2$, and $\kappa$ limits the overshoot caused by the last whole element in a phase.
	The prefix size $H$ is exactly the amount required by~\eqref{eq:prefixpayment}.
	Appendix~\ref{app:parameters} verifies the approximation and parameter orders.
	The grid is $\mathcal G=\{0,h,\ldots,h\floor{1/(2h)}\}$.
	
	\subsection{Prefix enumeration and continuation trees}
	Enumerate every feasible labeled restriction $\mathbf{y}$ of size at most $H$, including $\zero$.
	The term ``prefix'' refers to a restriction retained before the continuation; the enumeration need not enumerate different orderings of the same restriction.
	For each prefix, form $f_{\mathbf{y}}$ and residual normalized costs as in Section~\ref{sec:core}, using a fresh copy of the normalized input costs.
	Retain $\mathbf{y}$ as an output candidate.
	At the root of its continuation tree, process all elements with $\cost_e=\zero$ in a fixed order, sampling their top-2 labels once.
	
	To approximate the unknown weights, we guess their defining coefficients $\lambda_j$, rather than optimum memberships.
	At a phase beginning, a guessed vector and the known resource remainder determine frozen weights through~\eqref{eq:frozen}.
	We keep those weights until a resource has decreased by a prescribed factor; thus one guess may govern many element selections.
	More precisely, the positive-cost part is divided into phases as follows.
	At a phase beginning, record $\res^a=\res$ and branch over grid vectors $\widehat{\lvec}\in\mathcal G^d$ that do not exceed the preceding guessed vector componentwise.
	The initial upper vector is $(1/2,\ldots,1/2)$.
	To form positive resource weights on every branch, replace $\widehat\lambda_j$ by $\max\{\widehat\lambda_j,t/2\}$ in the following formula:
	\begin{equation}\label{eq:frozen}
		\widehat w_j=\frac{e^{-\widehat\lambda_j/\tau}}{\sum_{\ell=1}^d e^{-\widehat\lambda_\ell/\tau}},\qquad
		\widehat p_j=\widehat w_j\frac{\phi(\widehat\lambda_j)}{r_j^a}.
	\end{equation}
	The original grid coordinates, before clipping, are retained for the next monotonicity test.
	Only the resource weights remain fixed throughout the phase: the current labeling, marginals, scores, and resource remainders continue to change after each selection.
	Select candidates by~\eqref{eq:density}, sample their labels, and end the phase after a selection if some $r_j\le e^{-h}r_j^a$.
	Before every selection, stop if some $r_j\le\rho$; otherwise delete unsafe elements.
	Also stop if $U=\varnothing$ or $\max_{e\in U}G_e=0$.
	
	\begin{algorithm}[Enumerated Bellman greedy]\label{alg:outer}
		\Input{Value oracle $f$; $k\ge2$; $d$ nonnegative cost rows; $0<\varepsilon<1/2$.}
		\Output{A feasible labeling $\mathbf{s}^*$.}
		\State{Preprocess budgets and elements; compute~\eqref{eq:parameters}; $\mathcal C\gets\{\zero\}$.}
		\State{\textbf{for each} feasible $\mathbf{y}\in(k+1)^V$ with $|\supp(\mathbf{y})|\le H$ \textbf{do}}
		\StateIndent{1}{Let $V_{\mathbf{y}}\gets V\setminus\supp(\mathbf{y})$; exclude from $V_{\mathbf{y}}$ elements with positive cost in a zero residual-budget row.}
		\StateIndent{1}{Replace zero residual rows by unit-budget zero rows; normalize positive residual budgets; $f_{\mathbf{y}}(\mathbf{x})\gets f(\mathbf{y}\sqcup\mathbf{x})$.}
		\StateIndent{1}{$\mathbf{s}\gets\zero$, $U\gets V_{\mathbf{y}}$, $\res\gets\one$.}
		\StateIndent{1}{\textbf{for each} $e\in U$ with $\cost_e=\zero$ in fixed order \textbf{do}}
		\StateIndent{2}{$(\mathbf{q}_e,G_e)\gets\Top((\Delta_{e,i}f_{\mathbf{y}}(\mathbf{s}))_{i\in[k]})$; sample $i\sim\mathbf{q}_e$.}
		\StateIndent{2}{$\mathbf{s}\gets\mathbf{s}+i\one_e$; $U\gets U\setminus\{e\}$.}
		\StateIndent{1}{$\mathcal T\gets\operatorname{Continue}(\mathbf{s},U,\res,(1/2)\one)$ \quad (Algorithm~\ref{alg:tree}).}
		\StateIndent{1}{$\mathcal C\gets\mathcal C\cup\{\mathbf{y}\}\cup\{\mathbf{y}\sqcup\mathbf{x}:\mathbf{x}\in\mathcal T\}$.}
		\State{$\mathbf{s}^*\in\arg\max_{\mathbf{x}\in\mathcal C}f(\mathbf{x})$; \textbf{return} $\mathbf{s}^*$.}
	\end{algorithm}
	
	\begin{algorithm}[Continuation tree: $\operatorname{Continue}(\mathbf{s},U,\res,\mathbf{a})$]\label{alg:tree}
		\Input{A shared realized state; previous grid upper bounds $\mathbf{a}$; residual costs and $f_{\mathbf{y}}$.}
		\Output{A set of terminal labelings.}
		\State{\textbf{if} $\min_jr_j\le\rho$ \textbf{then return} $\{\mathbf{s}\}$.}
		\State{$U\gets\{e\in U:c_{j,e}\le\kappa r_j\ \forall j\}$; \textbf{if} $U=\varnothing$ \textbf{then return} $\{\mathbf{s}\}$.}
		\State{$\mathcal T\gets\varnothing$; $\res^a\gets\res$.}
		\State{\textbf{for each} $\widehat{\lvec}\in\mathcal G^d$ with $\widehat{\lvec}\le\mathbf{a}$ \textbf{do}}
		\StateIndent{1}{Copy the shared state to $(\mathbf{x},W,\mathbf{r}^{\rm cur})\gets(\mathbf{s},U,\res)$; freeze~\eqref{eq:frozen}.}
		\StateIndent{1}{\textbf{repeat}}
		\StateIndent{2}{\textbf{if} $\min_jr_j^{\rm cur}\le\rho$ \textbf{then} $\mathcal T\gets\mathcal T\cup\{\mathbf{x}\}$; \textbf{break}.}
		\StateIndent{2}{$W\gets\{e\in W:c_{j,e}\le\kappa r_j^{\rm cur}\ \forall j\}$.}
		\StateIndent{2}{\textbf{if} $W=\varnothing$ \textbf{then} $\mathcal T\gets\mathcal T\cup\{\mathbf{x}\}$; \textbf{break}.}
		\StateIndent{2}{$(\mathbf{q}_e,G_e)\gets\Top((\Delta_{e,i}f_{\mathbf{y}}(\mathbf{x}))_{i\in[k]})$ for all $e\in W$.}
		\StateIndent{2}{\textbf{if} $\max_{e\in W}G_e=0$ \textbf{then} $\mathcal T\gets\mathcal T\cup\{\mathbf{x}\}$; \textbf{break}.}
		\StateIndent{2}{$e\in\arg\max_{a\in W}G_a/(\hprices^{\mathsf T}\cost_a)$; sample $i\sim\mathbf{q}_e$.}
		\StateIndent{2}{$\mathbf{x}\gets\mathbf{x}+i\one_e$; $W\gets W\setminus\{e\}$; $\mathbf{r}^{\rm cur}\gets\mathbf{r}^{\rm cur}-\cost_e$.}
		\StateIndent{2}{\textbf{if} $r_j^{\rm cur}\le e^{-h}r_j^a$ for some $j$ \textbf{then}}
		\StateIndent{3}{$\mathcal T\gets\mathcal T\cup\operatorname{Continue}(\mathbf{x},W,\mathbf{r}^{\rm cur},\widehat{\lvec})$; \textbf{break}.}
		\State{\textbf{return} $\mathcal T$.}
	\end{algorithm}
	
	At each actual selection node, one label is sampled; the tree branches over coefficient guesses, not over all random labels.
	When a phase ends, every child inherits the same completed phase, including its sampled labels.
	Each child uses fresh randomness only for subsequent selections.
	The state is copied before exploring each child, so exploring one child does not alter another.
	Every branch is feasible: zero-cost selections consume no resource, and all later selections satisfy the safety test.
	All frozen resource weights are strictly positive because the clipped grid coordinates lie in $(0,1/2]$.
	Thus every positive-cost density denominator is positive.
	
	\section{Implementation, approximation, and query complexity}\label{sec:implementation}
	\subsection{Phasewise weight approximation}
	Fix the good prefix supplied by Lemma~\ref{lem:prefix}.
	For the analysis, follow the child obtained by rounding the true coefficients down to the grid at each phase beginning:
	\[
	\widehat\lambda_j=h\floor{\lambda_j^a/h}.
	\]
	These guesses are admissible because the true coefficients are nonincreasing.
	The algorithm explores every allowed child and does not need to identify this path.
	
	\begin{lemma}[Weight distortion]\label{lem:price}
		Along this path,~\eqref{eq:distortion} holds with $R=e^{t/2}$ whenever the comparison coefficient is positive.
	\end{lemma}
	\begin{proof}
		Consider a phase containing a selection with $\Lambda>0$.
		At its beginning $\lambda_j^a>t$ for every $j$ by~\eqref{eq:envelope}.
		The parameter choice gives $h<t/2$, so the downward-grid guess exceeds $t/2$ and clipping is inactive.
		Before a selection, $r_j>e^{-h}r_j^a$.
		The phase rule bounds resource use before that selection; the safety test bounds how far its last whole element can cross the phase boundary.
		Together with $1-\kappa=e^{-h}$, they imply throughout its resource interpolation
		\[
		r_j(u)\ge e^{-2h}r_j^a.
		\]
		By~\eqref{eq:relative}, the decrease of $\lambda_j$ from the phase beginning is at most
		\[
		\lambda_j^a(1-e^{-2h})\le\tfrac12(1-e^{-2h})\le h.
		\]
		Both $\lambda_j(u)$ and its downward-grid guess therefore lie in $[\lambda_j^a-h,\lambda_j^a]$, and
		\[
		|\lambda_j(u)-\widehat\lambda_j|\le h.
		\]
		If two vectors have coordinatewise error at most $h$, their log-sum-exp normalizers differ by at most $h/\tau$.
		Hence $|\log(w_j(u)/\widehat w_j)|\le2h/\tau$.
		All intervening scalar coefficients lie above $t/2$, so~\eqref{eq:condition} gives
		\[
		|\log(\phi(\lambda_j(u))/\phi(\widehat\lambda_j))|\le(2/t+c_Q)h.
		\]
		Finally $0\le\log(r_j^a/r_j(u))\le2h$.
		Adding the three errors yields
		\[
		\left|\log\frac{p_j(u)}{\widehat p_j}\right|
		\le h(2/\tau+2/t+c_Q+2)=t/2,
		\]
		as required.
		If $\Lambda$ reaches zero during a selection, the proof is used only up to that point; thereafter the comparison term vanishes.
	\end{proof}
	
	\subsection{Stopping and adaptive paths}
	The remaining two points are that the observable stopping rules remove the terminal comparison term, and that the well-approximated path is present in the sampled tree with the correct label distributions.
	\begin{lemma}[Terminal comparison]\label{lem:terminal}
		Every stopping rule of Algorithm~\ref{alg:tree} makes $\Lambda\pos{D}=0$ on the correct-prefix path.
	\end{lemma}
	\begin{proof}
		If $U=\varnothing$, there is no unselected, undeleted optimum element, so $\mathbf{z}=\mathbf{s}$.
		If all candidate scores are zero, the marginal-sum bound~\eqref{eq:marginalgap} gives $D\le0$, since all remaining comparison elements are candidates.
		If $r_j\le\rho$ for some $j$, then $m_j-r_j\ge t$ implies
		\[
		\lambda_j\le\frac{r_j}{2r_j+Qt}\le\frac{\rho}{2\rho+Qt}=t.
		\]
		The smooth minimum is at most $\lambda_j$, and~\eqref{eq:envelope} consequently gives $\Lambda=0$.
	\end{proof}
	
	\begin{lemma}[A correct path in the sampled tree]\label{lem:path}
		For each fixed optimum and correct prefix, the implemented tree contains an adaptive path satisfying the conditional distributions used in Lemma~\ref{lem:integer} and the weight condition of Lemma~\ref{lem:price} until $\Lambda$ becomes zero.
		Its terminal value is bounded above by the output value on every realization.
	\end{lemma}
	Here an adaptive path means that the next coefficient guess is determined from the already realized parent state, before revealing any new label on that child.
	The optimum is fixed only for this analysis; the algorithm itself still runs all children.
	\begin{proof}
		Use fresh independent randomness at each actual selection node; the tree may be generated in depth-first order.
		At a phase beginning on the comparison path, the current support and the fixed optimum determine $m_j$, $r_j$, and $\lambda_j^a$.
		Identify its downward-grid child from this information before exposing any random label below the phase node.
		The choice is admissible because $\lambda_j^a$ is nonincreasing.
		This choice depends only on the path history; randomness in other children is irrelevant and may be revealed in any order.
		Conditional on that history, each next previously unvisited selection node has its specified top-2 law.
		Repeating this construction gives the required conditional distributions at all steps.
		Once $\Lambda=0$, follow any admissible nonincreasing child, for example the zero vector.
		It remains zero, so only nonnegative actual gains are needed.
		All selections and deletions are finite, and the chosen terminal labeling is one of the algorithm's output candidates.
		Taking the best candidate therefore gives pointwise domination.
	\end{proof}
	Sharing the realized prefix is essential: independently rerunning complete guess sequences would resample their common early labels, so a later guess chosen from one run's state need not continue from that state in another run.
	The shared tree instead retains the realized parent state and includes every allowed continuation from it.
	
	\begin{theorem}[Main result]\label{thm:main}
		For every $k\ge2$, $d\ge1$, and $0<\varepsilon<1/2$, Algorithm~\ref{alg:outer} returns a feasible labeling on every outcome and satisfies
		\[
		\E f(\mathbf{s}^*)\ge(1/2-\varepsilon)\OPT.
		\]
		Its value-query complexity is at most
		\[
		\mathcal A(d,\varepsilon)(nk+1)^{H+2},\qquad
		H=O\!\left(\frac{d\log(2d)}{\varepsilon^2}\right),
		\]
		where $\mathcal A(d,\varepsilon)$ is independent of $n$ and $k$.
		For each enumerated prefix, its continuation uses at most $\mathcal B(d,\varepsilon)(n^2k+1)$ value queries, where $\mathcal B$ is also independent of $n,k$.
	\end{theorem}
	\begin{proof}[Proof of the approximation guarantee]
		Feasibility was established after Algorithm~\ref{alg:tree}.
		By Lemma~\ref{lem:prefix}, either an enumerated prefix already has value $\OPT$, or there is a size-$H$ prefix satisfying its marginal bound.
		For that prefix, Lemmas~\ref{lem:price}--\ref{lem:path} verify every condition of Proposition~\ref{prop:comparison}; the definition of $H$ verifies~\eqref{eq:prefixpayment}.
		Consequently,
		\[
		\E f(\mathbf{s}^*)\ge\frac{\Lambda_0}{K}\OPT
		=e^{-t}\left(\frac1{2+Qt}-t-\frac{t\log d}{1+\log d}\right)\OPT
		\ge(1/2-\varepsilon)\OPT,
		\]
		where the last inequality is proved in Appendix~\ref{app:parameters}.
		The query bound is proved next.
	\end{proof}
	
	\subsection{An explicit value-query bound}
	The bound must count both prefix enumeration and the sampled tree of coefficient guesses.
	Let
	\begin{equation}\label{eq:counts}
		M=1+d+\ceil{\frac d h\log\frac1\rho},\qquad
		J=1+\floor{\frac1{2h}},\qquad
		N=\binom{M+J}{J}^{d}.
	\end{equation}
	Here $M$ bounds the number of phase beginnings, $J$ counts grid levels per coordinate, and $N$ bounds the number of leaves.
	They depend only on $d,\varepsilon$; a phase may still contain up to $n$ selections.
	
	\begin{lemma}[Number of phases and coefficient traces]\label{lem:count}
		Every branch, including one with incorrect guesses, has at most $M$ phase beginnings.
		The continuation tree has at most $N$ leaves and at most $(n+M+1)N$ marginal-scan states and phase nodes, including terminal states.
	\end{lemma}
	\begin{proof}
		Before any positive-cost selection, $r_j>\rho$ for every row; immediately afterwards $r_j>e^{-h}\rho$ by safety.
		Thus the total logarithmic decrease in a row is at most $\log(1/\rho)+h$.
		Each completed phase decreases $\log r_j$ by at least $h$ in some row.
		Charging completed phases to witnessing rows bounds their number by $d\log(1/\rho)/h+d$; one further phase gives the bound $M$.
		This uses no comparison resource and therefore holds on every branch.
		
		A nonincreasing sequence of length $M$ over $J$ ordered levels has $\binom{M+J-1}{J-1}$ possibilities.
		Pad a terminal trace by repeating its last vector to length $M$.
		Two distinct leaves cannot receive the same padded trace: otherwise their guessed vectors agree up to the shorter trace, which reaches the same shared realized state; its termination prevents any distinct descendant leaf with that prefix.
		Random labels do not add branches, because each actual node samples one label.
		Leaves terminating before a first phase add at most the trivial one-leaf case.
		The conservative bound $N$ follows.
		A marginal scan computes the current scores of all candidates.
		Each path contains at most $n$ scans followed by a selection, one terminal scan or other terminal state, and $M$ phase beginnings, so the total node count is at most $(n+M+1)N$.
	\end{proof}
	
	\begin{proposition}[Query accounting]\label{prop:queries}
		Let $P=\sum_{i=0}^{\min\{H,n\}}\binom ni k^i$.
		The number of exact value-oracle queries is at most
		\begin{equation}\label{eq:queries}
			P\left[(n+M+1)N(nk+2)+nk+1\right].
		\end{equation}
		In particular, it is at most $\mathcal A(d,\varepsilon)(nk+1)^{H+2}$.
	\end{proposition}
	\begin{proof}
		At most $P$ labeled restrictions are considered before the feasibility test, hence at most $P$ feasible prefixes.
		At a marginal-scan state, including a terminal scan with zero scores, one base-value query and at most $nk$ additional queries compute every marginal needed by the density rule; one further query can evaluate a terminal candidate.
		Charging $nk+2$ queries at every counted node is an upper bound even for phase nodes requiring no queries.
		At the shared zero-cost root, begin with one query for $f(\mathbf{y})$ and query the $k$ successor values for each processed element.
		Cache the queried successor value corresponding to the sampled label as the next base value.
		This root therefore uses at most $nk+1$ queries.
		Retaining a prefix and comparing outputs is covered by these evaluations.
		This proves~\eqref{eq:queries}.
		Finally,
		\[
		P\le(H+1)(nk+1)^H,\qquad
		n+M+1\le(M+2)(nk+1),\qquad nk+2\le2(nk+1).
		\]
		Thus one may take $\mathcal A=(H+1)[2(M+2)N+1]$.
		For a fixed prefix,~\eqref{eq:queries} without $P$ is $\mathcal B(d,\varepsilon)(n^2k+1)$.
		Appendix~\ref{app:parameters} gives the stated order of $H$.
	\end{proof}
	
	The same enumeration is polynomial in $n,k$ for fixed $d,\varepsilon$ in an exact real-arithmetic model allowing exponential and logarithmic evaluation.
	Arithmetic for each marginal scan, cost test, and density computation introduces a polynomial factor in $n,k,d$; the phase branching factors depend only on $d,\varepsilon$.
	The displayed bound is a value-query bound and makes no assumption on the bit encoding of oracle answers.

	\section{Conclusion}
	We extend the single-resource Bellman analysis to a fixed number of knapsacks by combining the row coefficients and choosing resource weights that control the decrease of the resulting common coefficient.
	Greedy gains pay for this decrease and the local top-2 comparison loss.
	A short optimum prefix pays for the additional loss from deleting unsafe optimum elements.
	Enumerating prefixes and monotone coefficient guesses in resource-based phases gives a combinatorial $(1/2-\varepsilon)$-approximation, without a lower bound on positive input costs.
	The guarantee approaches the optimal oracle threshold for bisubmodular objectives and, up to arbitrarily small $\varepsilon$, is asymptotically tight as $k\to\infty$.
	
	An important direction is to retain this resource accounting while replacing prefix or coefficient-trace enumeration by a more efficient combinatorial procedure.
	It is also natural to ask whether the approximation loss can be removed for a fixed number of knapsacks, or whether different comparison resources yield stronger guarantees for fixed $k\ge3$ and for additional support constraints.
	
	\appendix
	\section{Proof of the top-2 certificate}\label{app:top2}
	Relabel the largest two coordinates of $\mathbf{g}$ as $a=g_1\ge c=g_2\ge g_i$ for $i\ge3$.
	Pairwise monotonicity implies $a\ge0$.
	If $a=0$, all entries of $\mathbf{g}$ are zero, and any pairwise-monotone $\mathbf{v}\le\mathbf{g}$ is also zero.
	All three inequalities follow.
	If $c\le0<a$, then $q_1=1$ and $G=a$.
	For $i\ne1$, $v_i\le c\le0$ and $v_1\ge-v_i$, giving $v_i-v_1\le0$.
	Also $v_1\ge-v_2\ge-c\ge0$.
	Thus the replacement and addition bounds hold, and $a\le QG$.
	
	It remains to consider $a\ge c>0$, where
	\[
	q_1=\frac a{a+c},\qquad q_2=\frac c{a+c},\qquad G=\frac{a^2+c^2}{a+c}.
	\]
	For label $1$, $v_2\ge-v_1$ and $v_1\le a$ give
	\[
	v_1-\mathbf q^{\mathsf T}\mathbf v=q_2(v_1-v_2)\le2q_2a\le G.
	\]
	For label $2$, use $v_1\ge-v_2$ and $v_2\le c$ to get the same bound $2q_1c\le G$.
	For $i\ge3$, put $u=v_i\le c$.
	If $u<0$, then $v_1,v_2\ge-u$, so $u-\mathbf q^{\mathsf T}\mathbf v\le2u<0$.
	If $u\ge0$, write
	\[
	\mathbf q^{\mathsf T}\mathbf v=(q_1-q_2)v_1+q_2(v_1+v_2)\ge-(q_1-q_2)u,
	\]
	using $v_1\ge-u$ and $v_1+v_2\ge0$.
	Therefore $u-\mathbf q^{\mathsf T}\mathbf v\le2q_1u\le2q_1c\le G$.
	This proves the replacement bound for every label.
	The same decomposition and $v_1\ge-v_2\ge-c$ imply
	\[
	-\mathbf q^{\mathsf T}\mathbf v\le\frac{c(a-c)}{a+c}.
	\]
	Since $A(A+1)=1/4$ and $Q=A+1$,
	\[
	A(a^2+c^2)-c(a-c)=(\sqrt A\,a-\sqrt{A+1}\,c)^2\ge0,
	\]
	which proves the addition bound.
	Finally,
	\[
	Q(a^2+c^2)-a(a+c)=(\sqrt A\,a-\sqrt Q\,c)^2\ge0,
	\]
	proving $a\le QG$ and completing Lemma~\ref{lem:top2}.
	
	\section{Accuracy and parameter orders}\label{app:parameters}
	The parameter choices come from the weight-error and prefix-payment identities in~\eqref{eq:parameters}.
	Here we verify that they give the required approximation and query exponent.
	Since $0<\varepsilon<1/2$,
	\[
	0<t<\frac1{5+Q/2}<\frac12.
	\]
	Moreover,
	\[
	\frac1{2+Qt}\ge\frac12-\frac{Qt}{4},\qquad
	\tau\log d\le t,
	\]
	so
	\[
	\Lambda_0\ge\frac12-(2+Q/4)t>0.
	\]
	The strict positivity follows by substituting $t<1/(5+Q/2)$.
	Also $\Lambda_0\le1/2$, and $e^{-t}\ge1-t$ yields
	\[
	\frac{\Lambda_0}{K}=e^{-t}\Lambda_0
	\ge\Lambda_0-t/2
	\ge\frac12-(5/2+Q/4)t=\frac12-\varepsilon.
	\]
	The positive initial coefficient implies $t<1/(2+Qt)$, equivalently $Qt^2<1-2t$; hence $0<\rho<1$.
	
	To prove~\eqref{eq:condition}, differentiate $\log\phi$ and observe that
	\[
	(1-b\lambda)(1-Q\lambda)\ge(1-b/2)(1-Q/2)=Qb/4
	\]
	for $0<\lambda\le1/2$; $Q-b=2(Q-1)$ gives the displayed $c_Q$.
	Because $2/\tau+2/t+c_Q+2\ge4/t$, $h\le t^2/8<t/2$.
	Conversely $t$ is bounded above by an absolute constant and
	\[
	h=\Theta\!\left(\frac{t^2}{1+\log d}\right),\qquad
	\kappa=1-e^{-h}=\Theta(h).
	\]
	Finally $K-\Lambda_0\ge1/2$ and $K-\Lambda_0\le e^t$, so
	\[
	H=\Theta\!\left(\frac{d(1+\log d)}{t^2}\right)
	=O\!\left(\frac{d\log(2d)}{\varepsilon^2}\right).
	\]
	All quantities in~\eqref{eq:counts} consequently depend only on $d,\varepsilon$.
	This parameter analysis requires no lower bound on a positive input cost.
	
	
	\section*{AI assistance and author responsibility}
	The main ideas and methods of this paper build on the authors' earlier manuscript~\cite{chen2026}.
	The algorithm and analysis were developed through repeated interactions between the authors and GPT-6.1 Sol Ultra.
	
	First, the authors asked GPT to combine the single-knapsack Bellman analysis in~\cite{chen2026} with multilinear-extension techniques to obtain a $(1/2-\varepsilon)$-approximation under multiple knapsack constraints.
	After several rounds of interaction, GPT produced an algorithm and analysis.
	The authors then asked for a combinatorial algorithm with the same approximation guarantee.
	They pursued this task in parallel with GPT-6 Pro in chat mode and GPT-6.1 Sol Ultra in work mode.
	The two sessions had no shared memory, although the authors manually exchanged their intermediate results.
	GPT-6.1 Sol Ultra ultimately produced the combinatorial algorithm and analysis presented here.
	
	GPT generated the initial draft of the manuscript.
	The authors reviewed and verified the algorithm and its correctness proofs and made preliminary editorial revisions.
	They take responsibility for the content of the paper.
	The manuscript remains a draft: its readability still requires improvement, and its exposition will be further revised in subsequent versions.
	
	
	\begingroup
	\raggedright

	\endgroup
\end{document}